\documentclass[letterpaper, 10 pt, conference]{ieeeconf}  

\IEEEoverridecommandlockouts                              
\usepackage{cite}
\usepackage{amsmath,amssymb,amsfonts,verbatim}
\usepackage{algorithmic}
\usepackage{graphicx}
\usepackage[utf8]{inputenc}
\usepackage{textcomp}
\usepackage{xcolor}
\usepackage{braket}
\usepackage{soul}
\usepackage{bbm}
\usepackage[colorinlistoftodos]{todonotes}
\newtheorem{theorem}{Theorem}

\DeclareMathOperator{\Tr}{Tr}
\usepackage{tikz}
\usetikzlibrary{decorations.pathmorphing} 
\usetikzlibrary{matrix} 
\usetikzlibrary{arrows} 
\usetikzlibrary{calc} 

\tikzstyle{block} = [draw,align=center,rectangle,thick,minimum height=4em,minimum width=4em,text width=3cm]
\tikzstyle{state} = [draw,align=center, circle,inner sep=0mm,minimum size=1.7cm,text width=1cm]
\tikzstyle{connector} = [->,thick]
\tikzstyle{line} = [thick]
\tikzstyle{guide} = []

\newcommand{\EX}{\mathbb{E}}
\newcommand{\PR}{\mathbb{P}}
\newcommand{\BN}{\mathbb{N}}
\newcommand{\CX}{\mathcal{X}}
\newcommand{\CA}{\mathcal{A}}
\newcommand{\CS}{\mathcal{S}}
\newcommand{\CH}{\mathcal{H}}

\newcommand{\CL}{\mathcal{L}}
\newcommand{\complex}{\mathbb{C}}
\newcommand{\ly}{Lyapunov }
\newcommand{\limset}{$D_{\infty}$}
\DeclareMathOperator{\vect}{vec}
\newcommand{\state}{s}
\newcommand{\belief}{\pi}
\newcommand{\id}{\mathbb{I}}
\newcommand{\diag}{\text{diag}}
\newcommand{\reals}{{\rm I\hspace{-.07cm}R}}
\def\BibTeX{{\rm B\kern-.05em{\sc i\kern-.025em b}\kern-.08em
    T\kern-.1667em\lower.7ex\hbox{E}\kern-.125emX}}

\title{\LARGE \bf
\ly based  Stochastic Stability of  Human-Machine Interaction: \\  A Quantum Decision System Approach
}

\author{{Luke Snow}, {Shashwat Jain},  {Vikram Krishnamurthy}
\thanks{This research was supported in part by the National Science Foundation grant CCF-2112457 and Army Research Office grant W911NF-19-1-0365}
\thanks{Luke Snow {\tt\small las474@cornell.edu}, Shashwat Jain {\tt\small sj474@cornell.edu}, Vikram Krishnamurthy {\tt\small vikramk@cornell.edu} are  with the School of Electrical and Computer Engineering, Cornell University, Ithaca, NY 14853, USA
        }%
}

\begin{document}

\maketitle

\begin{abstract}
In mathematical psychology,  decision makers are modeled using the Lindbladian equations from quantum mechanics to capture important human-centric features such as  order effects and violation of the sure thing principle. We consider   human-machine interaction involving a quantum decision maker (human)  and a controller (machine). Given a sequence of  human decisions over time,  how can the controller dynamically provide input messages to  adapt these decisions so as to  converge to a specific decision?  We show via  novel stochastic \ly arguments  how  the Lindbladian dynamics of the quantum decision maker can be controlled to converge to a specific decision asymptotically.  Our methodology yields a useful mathematical framework for human-sensor decision making. The stochastic \ly results are also of independent interest as they generalize recent results in the literature.

\end{abstract}

\section{Introduction}

Recent studies in mathematical psychology \cite{martinez2016quantum}, \cite{kvam2021temporal}, \cite{busemeyer2020application}, show that the Lindbladian equations from quantum mechanics facilitate modeling peculiar aspects  of human decision making. Such quantum decision models capture {\em order effects}
(humans  perceive $P(H|A\cap B) $ and $P(H|B \cap A)$ differently in decision making)
and violation of the {\em sure thing principle} (human perception of probabilities in decision making violates the total probability rule). Motivated by the design of human-machine interaction systems, this paper addresses the following question: {\em Given a sequence of human decisions over time, how can a controller (machine)  adapt the Lindbladian dynamics  (of the human decision maker) so as to converge to a specific decision?} To investigate this, we develop a novel generalization of recent results involving finite-step stochastic \ly functions. Thus at an abstract level, we study the stochastic stability of a switched controlled Lindbladian dynamic system where the switching occurs due to the interaction of the controller (machine)  and decision maker (human) at specific (possibly random) time instants. 

\subsection{Decision Making Context}
Figure~\ref{fig:Model} shows our schematic setup. The finite-valued random variable $\state\sim \belief_0(\cdot)$ 
denotes the underlying state of nature, where $\belief_0$ is a known probability mass function.
The input signals $y_k$ and $z_k$ are noisy observations of the state with conditional observation densities $p(y|\state)$ and $p(z|\state)$, respectively. The human's psychological state $\rho$ is represented as a density operator in Hilbert Space, which evolves via the Lindbladian equation parametrized by the observation $y_k$ and input $u_k$. The density operator $\rho$ encodes a probability distribution over actions $\{a_j\}_{j \in \{1,\dots,m\}}$, and at each time point an action is taken according to this distribution. The machine observes the actions and outputs a feedback control signal $u_k$ to the human.


\subsubsection*{Examples} Several
examples in robotics \cite{askarpour2019formal},  interactive marketing/advertising \cite{belanche2020consumer} and  recommender systems \cite{lu2012recommender} exploit models for human decision making.
One example is a machine assisted healthcare system for patients with dementia \cite{HOEY2010503}, in which the patient is assisted by a machine (smart watch) to wash his hands. The machine's sensor detects whether a certain set of sequential actions are followed by the patient, and then sends those results to a controller which gives an audio/video command to the patient. 
In this context the underlying states ($s$) are the tap water, soap dispenser and towel dispenser, which are partially observed by both sensor and the patient. The patient has a psychological state ($\rho$) and the resultant hand washing actions ($a_k$) are sensed by the sensor, then the controller gives the control input ($u_k$). In our work, we model the psychological state of the patient as a Lindbladian evolution as shown in Figure~\ref{fig:Model} since this accounts for a wider range of human behavior, such as irrational decisions which could be made by the dementia patient, than classical models.

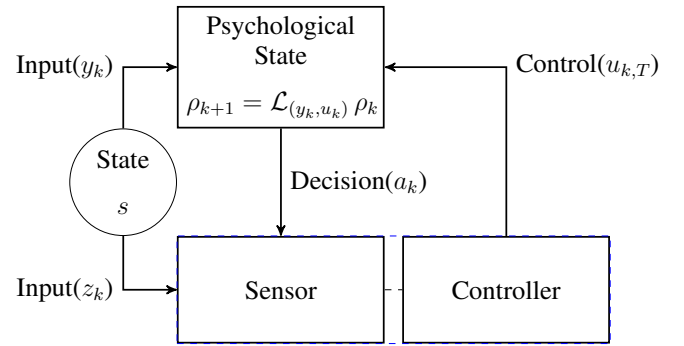
\begin{figure}[h!]
    \centering
    \resizebox{0.5\textwidth}{!}{%
  \begin{tikzpicture}[scale=1, auto, >=stealth']
     \large
    \matrix[ampersand replacement=\&, row sep=0mm, column sep=0cm] {
    
      \&\&\&\&\\
      \&\&\node[block] (F1) {Psychological State \[\rho_{k+1}=\mathcal{L}_{(y_k,u_k)}\,\rho_k\]};\&\node[connector] (u1) {};
      \&
      \&\\
      \node[state] (X){State\[s\]};\&\&\&\&\\
      \&\&\node[block] (F2) {Sensor};\&\&\node[block] (f1) {Controller};
      \\
    };
    \draw [dashed] (F2) -- (f1);
    \draw [connector] (F1) -- node {Decision($a_{k}$)} (F2);
    \draw [connector] (f1.north) -- ++(0,0cm) |- node [right]
                      {Control($u_{k,T}$)} ($(F1.east) + (0cm, 0em)$);
    \draw [connector] ($(X.north) - (0cm, 0cm)$) -- ++(0,0cm) |- node [left]
      {Input($y_k$)} ($(F1.west) - (0cm, 0em)$);
    \draw [connector] ($(X.south) - (0cm, 0em)$)-- ++(0,0cm) |- node [left]
      {Input($z_k$)} ($(F2.west) - (0cm, 0em)$); 
    \draw[dashed,blue] (F2.north west) rectangle (f1.south east);
  \end{tikzpicture} 
  }
    \caption{Human-Machine Interaction Model}
    \label{fig:Model}
\end{figure}
\subsection{Main Results and Organization}
    Given the described human-machine decision making system, the question we ask is: Can the human's decision preference be guided by the input control signals such that a desired target action is eventually taken at every time step? Our results reveal that this indeed is the case. We show this by developing a novel \ly stability result for a Lindbladian dynamic system.
    
    The main results and organization of the paper is as follows:\\
    (1) Sec. II introduces the open-quantum cognitive model of Martinez et. al. \cite{martinez2016quantum} and represents the discretized process in a form that will be mathematically useful for us.
    \\
    2) Sec. III presents Theorem (\ref{thm:convg}), which is our main result and shows the stochastic stability of our human-machine decision system. The proof uses the methodology of Amini et. al \cite{AMINI20132683}, along with Lyapunov techniques and Theorem (\ref{thm:extension}) which we provide in Sec. IV.\\
    3) Sec. IV provides a generalization of a finite-step \ly stability result given in Qin et. al. \cite{qin2019lyapunov} in Theorem (\ref{thm:rand}), to the case when the finite-step interval $T$ is a random variable. Also Theorem (\ref{thm:extension}) is a modified form of this result which is used to prove Theorem (\ref{thm:convg}) in Sec. III.

\subsection{Literature Review}
Generative models for human decision making are studied extensively in behavioral economics and psychology.
 The classical  formalisms  of human decision making are the Expected Utility models of Von-Neumann and Morgenstern (1953)\cite{morgenstern1953theory} and Savage (1954) \cite{savage1951theory}. Despite the successes of these models, numerous experimental findings, most notably those of Kahneman and Tverksy \cite{kahneman1982judgment}, have demonstrated violations of the proposed decision making axioms. There have since been subsequent efforts to develop axiomatic systems which encompass wider ranges of human behavior, such as the Prospect Theory \cite{kahneman2013prospect}. However, given the complexity of human psychology and behavior it is no surprise that current models still have points of failure. The theory of Quantum Decision Making (\cite{busemeyer2012quantum}, \cite{khrennikov2010ubiquitous}, \cite{yukalov2010mathematical} and references therein) has emerged as a new paradigm which is capable of generalizing current models and accounting for certain violations of axiomatic assumptions. For example, it has been empirically shown that humans routinely violate Savage's 'Sure Thing Principle' \cite{khrennikov2009quantum}, \cite{aerts2011quantum}, which is equivalent to violation of the law of total probability, and that human decision making is affected by the order of presentation of information \cite{trueblood2011quantum} \cite{busemeyer2011quantum} ("order effects"). These violations are natural motivators for treating the decision making agent's mental state as a quantum state in Hilbert Space; The mathematics of quantum probability was developed as an explanation of observed self-interfering and non-commutative behaviors of physical systems, directly analogous to the findings which Quantum Decision Theory (QDT) aims to treat. 
 
 {\em Remark}. QDT  models in  psychology do not claim that the brain is acting as a quantum device in any physical sense. Instead QDT  serves as a {\em parsimonious generative blackbox model} for human decision making that is backed up by extensive experimental studies \cite{kvam2021temporal}, \cite{busemeyer2012quantum}.

Within Quantum Decision Theory, several recent advances have utilized quantum dynamical systems to model time-evolving decision preferences. The classical model for this type of time-evolving mental state is a Markovian model, but in \cite{busemeyer2009empirical} an alternative formulation based on Schr\"{o}dinger's Equation is developed. This model is shown to both reconcile observed violations of the law of total probability via quantum interference effects and model choice-induced preference changes via quantum projection. This is further advanced in \cite{asano2012quantum}, and  \cite{martinez2016quantum} where the mental state is modeled as an open-quantum system. This open-quantum system representation allows for a generalization of the widely used Markovian model of preference evolution, while maintaining these advantages of the quantum framework. Busemeyer et. al. \cite{kvam2021temporal} provide empirical analysis which supports the use of open-quantum models and conclude "An open system model that incorporates elements of both classical and quantum dynamics provides the 
best available single system account of these three characteristics—evolution, oscillation, and choice-induced 
preference change".

\subsection*{Notation}
\begin{enumerate}
    \item $\ket{n} \in \mathcal{S}$ is $n^{th}$ basis vector in a Hilbert space $\mathcal{S}$.
    \item $\mathcal{H}(\{\ket{\state_1}, \dots, \ket{\state_k}\})$: Hilbert Space spanned by the orthonormal basis vector set $\{\ket {x_1},\dots,\ket {x_k}\}$ 
    \item $A^{\dagger}$: adjoint of $A$
    \item $\ket{n}\bra{m}$: outer product of $\ket{n}$ and $\ket{m}$
    \item Density operator $\rho$: $\rho=\ket{\Psi}\bra{\Psi}$ for some $\ket{\Psi}\in \CS$
    \item Trace of operator $A$: $\Tr(A)=\sum_{l=1}^{N}\bra{n}A\ket{n}$ for basis $\ket{n}\in\CS$
    \item  Random events are defined on $(\Omega,\mathcal{F},\PR)$.
\end{enumerate}
\section{Quantum Probability Model for Human Decision Making}
\label{sec2}
This  section presents the open-quantum system model that we will  use to represent the decision preference evolution of the human decision maker.
We define the evolution of the density operator of the decision maker using the open-system Quantum Lindbladian  Equation as given in \cite{martinez2016quantum}.

\subsection{Lindbladian Dynamics}
Given a state of nature which is a random variable that can take on $n$ values, the human decision maker chooses one of $m$ possible actions. 
The quantum based model for human decision making is governed by the Lindbladian evolution of the psychological state. With $\mathcal{L}$ denoting the Lindblad operator, the Lindbladian ordinary differential  equation for the dynamics of the psychological state $\rho_t$ over time $t \in [0,\infty)$  is specified as  
\begin{equation}
\frac{d\rho_t}{dt} = \mathcal{L}_{(\alpha,\lambda,\phi)} \, \rho_t, \quad \rho_0  = \frac{1}{nm}\diag(1,\cdots,1)_{mn\times mn} \label{eq:lindbladian}  \end{equation} 
Here $(\alpha,\lambda,\phi)$ are free parameters which determine the quantum decision maker's behavior, as discussed in Appendix~\ref{constr_lind}. We  assume  that the machine has full knowledge of these behavioral parameters, as methods for estimating these via training are outside the scope of this paper.

In subsequent sections we will control the evolution of $\rho_t$ and formulate \ly stability conditions. We are interested in the general case when the machine can only observe the human's actions and output a control every $T$ time steps, where $T$ is a random variable. We thus define the evolution of the density operator in  (\ref{M_update}), from the controller's perspective, over $T$-step iterations.

The psychological model comprises the following: 
\begin{enumerate}
\item  A state of nature $s \in \{1,\ldots,n\} = \CX$, with probability mass function $\pi_0(s)$
\item  An action $a_k \in \{1,\ldots,m\} = \CA$ by the human at discrete time~$k$, for $k = 0,1,2,\ldots$. 
\item  Scalar control input $u_k \in [-\bar{u},\bar{u}],\, \bar{u} \in \reals_+$, from the machine. This controls the parameters of the Lindbladian operator in equation (\ref{M_update}), and models a recommendation signal (for example a posterior probability of the state of nature) by the machine to the human.
\item The discrete time evolution of the psychological state
\begin{equation}
\label{M_update}
    \rho_{k+T} = \mathbb{M}_{\mu_k, T}^{u_k}\,\rho_k = \frac{M_{\mu_k, T}^{u_k}\, \rho_k\, M_{\mu_k, T}^{u_k\dagger}}{\Tr(M_{\mu_k, T}^{u_k}\, \rho_k\, M_{\mu_k, T}^{u_k\dagger})}
\end{equation} 
where $k = 0,1,\ldots$ denotes discrete-time.
Recall the random variable $T: \Omega \rightarrow \mathbb{N}$ specifies the time intervals over which the machine interacts with the human. $T$ has a known probability mass function $\pi_T(\cdot)$.  $\mu_k$ is a $T$-length sequence $\{\mu_{k_i}\}_{i = 1}^T $of random actions $\mu_{k_i}$ taking values in $\mathcal{A}$. 
See Appendix~\ref{DT_lind} for the definition of $\mathbb{M}_{\mu_k, T}^{u_k}$.
\item The action probability distribution at  time $k$
\begin{equation}
\label{eq:ProbEvo}
 p(a_k = \mu_k)=\Tr(M_{\mu_k}^{u_k}\, \rho_k \, M_{\mu_k}^{u_k\dagger}), \mu_k \in \{1,\cdots,m\}   
\end{equation}
\end{enumerate}
See Appendix~\ref{constr_lind}  and~\ref{DT_lind} for further model details.

\subsection{Practicality in Modeling Human Decision Making}
The above Lindbladian model
captures important human decision features such as the sure-thing principle and order effects, which we now describe.
These features cannot  be explained by purely Markovian models without sacrificing their explanatory power.

\subsubsection*{The violation of the sure-thing principle}
The total probability law, also called the Sure Thing Principle (STP), is  
\[P(A)=P(B)\,P(A|B)+P(\bar{B})\,P(A|\bar B)\]
for events $A$ and $B$. Violation occurs when $=$ is replaced by either $<$ or 
$>$.
Suppose $P(A|B) = 0.6$,  $P(A|\bar{B})= 0.5$. Then if the probability the human decision maker chooses action $A$ is either greater than 0.6 or less than 0.5, then the STP has been violated. 

 \cite{martinez2016quantum}  shows that the  Lindbladian model accounts for violations of the STP. Pothos and Busemeyer \cite{pothos2009quantum} (see also \cite{khrennikov2009quantum}) review  empirical evidence for the violation of STP and show how quantum models can account for it by introducing quantum interference in the probability evolution. 

\subsubsection*{Order Effects}
It is well-established in psychology  \cite{shanteau1970additive}, \cite{furnham1986robustness}, \cite{walker1972order} that the order of presented information can affect the final judgement of a human \cite{hogarth1992order}. Order effects are not easily accounted for using classical set-theoretic probability axioms, since  $P(H | A \cap B) = P(H| B \cap A)$, i.e. the order of presentation of events $A$ and $B$ does not influence the final probability judgement of $H$. Alternative models of inference have been proposed, such as the averaging model \cite{shanteau1970additive} and the belief-adjustment model \cite{hogarth1992order}, but these are only heuristic ad-hoc models which lack axiomatic foundations. Quantum probability is a natural axiomatic framework which can account for these effects, see \cite{trueblood2011quantum}, \cite{busemeyer2012quantum}, \cite{khrennikov2010ubiquitous}, \cite{busemeyer2011quantum} and references therein. Order effects naturally arise from the non-commutative structure of quantum interactions.

\section{Machine Control of Human Decision Maker}
\label{sec3}
This section  exploits the \ly function formulated  in \cite{AMINI20132683} and a generalized finite-step convergence theorem, (that will be proved in Section~\ref{sec4}), to prove our main result, Theorem \ref{thm:convg}. This theorem states that regardless of the initial psychological state of the human, the machine is  able to control the preference in such a way that the target action is eventually chosen at every time step with probability one.

We first define some notation:
With $ d =nm$, let $D$ denote the space of non-negative Hermitian matrices with trace~1:
\begin{equation}
D := \{\rho \in \complex^{d \times d} : \rho = \rho^{\dagger}, \Tr(\rho) = 1, \rho \geq 0\}
\end{equation}
 Let $\{\ket{b_r}\}_{r=1}^d$ be a set of orthonormal vectors in $\complex^d$, where each $\ket{b_r}$ corresponds to a unique state-action pair. Let $\CS$ be the Hilbert space formed by taking these $\{\ket{b_r}\}_{r=1}^d$ as an orthonormal basis. We consider scalar control inputs $u_k \in \reals$ satisfying constraints given in Appendix \ref{ap:contr}. For our purposes it suffices that $u_k \in [-1,1]$, see \cite{AMINI20132683} for details. The term {\em 'Open-loop (super) martingale'} below denotes a (super) martingale when  the control input $u_k = 0$ for $k = 0,1,\ldots$. 
 
 The following is the main result:\\

\begin{theorem}
\label{thm:convg}
    Given the discrete time density operator evolution (\ref{M_update}) and any target state $\ket{\bar{b}_r}, r \in \{1,\dots,d\}$, there exists a control sequence $\{u_k\}_{k \in \mathbb{N}}$ generated by the machine such that the human  psychological state $\rho_k$ converges to $\ket{\bar{b}_r}\bra{\bar{b}_r}$ with probability one for any initial $\rho_0 \in D$.\\
\end{theorem}
\begin{proof}
We will follow the  formulation developed in \cite{AMINI20132683}. First
with $\beta_k = \{\mu_k, T\}$, rewrite (\ref{M_update}) as 
\begin{equation}
\label{eq:dens_up}
    \rho_{k+T} = \mathbb{M}_{\beta_k}^{u_k}\rho_k = \frac{M_{\beta_k}^{u_k}\, \rho_k\, M_{\beta_k}^{u_k \dagger}}{\Tr(M_{\beta_k}^{u_k}\, \rho_k\, M_{\beta_k}^{u_k \dagger})}
\end{equation}
 We define the following \ly function, which forms a supermartingale under both open-loop (zero-input) and closed-loop (feedback control ($u_k$)) conditions for the process (\ref{eq:dens_up}):
\begin{equation}
\label{Ly_fcn}
    V_{\epsilon}(\rho):= \sum_{r=1}^d \sigma_r \bra{b_r} \rho \ket{b_r} - \frac{\epsilon}{2}\sum_{r=1}^d \bra{b_r} \rho \ket{b_r}^2
\end{equation}
where $\sigma_r$ is non-negative $\forall r \in \{1,\ldots,d\}$ and $\epsilon$ is strictly positive. The set $\{\sigma_r\}_{r=1}^d$ and $\epsilon$ are chosen according to \cite{AMINI20132683} such that the \ly function $V_{\epsilon}(\rho_k)>0 \, \forall\, \rho_k \in D$ and $\rho_k$ converges to the intended subspace $\ket{\bar{b}_r} \bra{\bar{b}_r}$ with probability 1. By \cite{AMINI20132683} and \cite{6160433}, $\bra{b_r} \rho \ket{b_r}$ is an open-loop martingale given the density operator evolution (\ref{eq:dens_up}) (see Appendix~\ref{ap:martingale} for proof). $V_{\epsilon}$ is a concave function of the open-loop martingales $\bra{b_r} \rho \ket{b_r}$ and therefore is an open-loop ($u_k = 0$) supermartingale given the process (\ref{eq:dens_up}). 
\begin{equation*}
    \EX[V_{\epsilon}(\rho_{k+T}) |\, \rho_k, u_k = 0] - V_{\epsilon}(\rho_k) \leq 0 
\end{equation*}
The following feedback control mechanism is used
\begin{equation*}
    u_k := \underset{u \in [-1,1]}{\text{argmin}}\, \EX[V_{\epsilon}(\rho_{k+T}) | \rho_k, u] 
\end{equation*}
to get $\EX[V_{\epsilon}(\rho_{k+T}) | \rho_k, u_k] \leq \EX[V_{\epsilon}(\rho_{k+T}) | \rho_k, u = 0]$.
Here the expectation is taken with respect to $\beta_k$,
\begin{equation*}
\begin{split}
     \EX[V_{\epsilon}(\rho_{k+T}) | \rho_k, u]  &=
    \EX[V_{\epsilon}(\mathbb{M}_{\beta_k}^{u} \rho_k)]
=  \int_{\Omega} V_{\epsilon} (\mathbb{M}_{\beta_k(\omega)}^u \rho_k)d\omega 
\end{split}
\end{equation*}
where $\Omega$ is the sample space under which the process is induced.
Define 
$
    \tilde{Q}(\rho_k) := \EX[V_{\epsilon}(\rho_{k+T}) | \rho_k, u_k] - V_{\epsilon}(\rho_k) \leq  \EX[V_{\epsilon}(\rho_{k+T}) | \rho_k, u_k = 0] - V_{\epsilon}(\rho_k) \leq 0
$.
$V_{\epsilon}(\rho)$ is a continuous function, so using \cite[Chapter 8]{kushner1971introduction}, we have the $T$-step control sequence $\{\rho_{k+iT}\}_{i \in \mathbb{N}}$ that converges to the set $\textrm{\limset}:= \{\rho : \tilde{Q}(\rho) = 0\}$ with probability one.\\ We will first show that the set \limset \ is restricted to our desired state $\{\ket{\bar{b}_r} \bra{\bar{b}_r}\}$, then that the entire sequence $\{\rho_k\}_{k \in \mathbb{N}}$ converges to this set. The former is proved in Lemma 2 of \cite{AMINI20132683}; For any target subspace $\ket{\bar{b}_r} \bra{\bar{b}_r}$, the set $\{\sigma_r\}_{r=1}^d$ can be chosen in such a way that \limset $ = \ket{\bar{b}_r} \bra{\bar{b}_r}$. The idea is the following:
A state $\rho_k$ is in the limit set \limset \ iff for all $u \in [-1,1],$
\begin{equation}
\label{supm}
    \EX[V_{\epsilon}(\rho_{k+T}) | \rho_k, u] - V_{\epsilon}(\rho_k) \geq 0
\end{equation}
Also, since $V_{\epsilon}$ is an open-loop supermartingale, $\forall \rho_k \in D$: 
\begin{equation}
\label{subm}
    \EX[V_{\epsilon}(\rho_{k+T}) | \rho_k, u = 0] - V_{\epsilon}(\rho_k) \leq 0
\end{equation}
By Lemma 2 of \cite{AMINI20132683}, given any $\bar{r} \in \{1,\dots, d\}$ and fixed $\epsilon > 0$, the weights $\sigma_1, \dots, \sigma_d$ can be chosen so that $V_{\epsilon}$ satisfies the following property: $\forall r \in \{1,\dots,d\}, u \mapsto \EX[V_{\epsilon}(\rho_{k+T}) \ | \rho_k = \ket{b_r}\bra{b_r}, u_k = u]$ has a strict local minimum at $u=0$ for $r = \bar{r}$ and a strict local maximum at $u=0$ for $r \neq \bar{r}$. This combined with (\ref{subm}) ensures that for any $r \neq \bar{r}, \exists\, u \in [-1,1]$ such that $\EX[V_{\epsilon}(\rho_{k+T}) \ |\, \rho_k = \ket{b_r}\bra{b_r}, u_k = u] - V_{\epsilon}(\ket{b_r}\bra{b_r}) < 0$. Therefore, by (\ref{supm}), we have that $\ket{b_r}\bra{b_r}$ is in the limit set $l_{\infty}$ if and only if $r = \bar{r}$. 

We now show that $\PR[\lim_{k \to \infty}\rho_k = \ket{\bar{b_r}}\bra{\bar{b_r}}] = 1$, i.e. the entire sequence converges to the target state with probability one. This utilizes Theorem \ref{thm:extension} which is developed in Section \ref{sec4}.
We have shown that 
$
\PR[\,\lim_{i \to \infty} \rho_{k + iT} = \ket{\bar{b}_r}\bra{\bar{b}_r}] = \\
\PR[\,\exists\,K \in \mathbb{N} : \rho_{k + iT} = \ket{\bar{b}_r}\bra{\bar{b}_r} \ \forall i \geq K]=1    
$.
The set $\{\sigma_r\}$ was chosen such that $\EX[V_{\epsilon}(\rho_{k+T}) \ | \rho_k = \ket{b_r}\bra{b_r}, u_k = u]$ has a strict local minimum at $u=0$ for $r = \bar{r}$ so
$
    u_{k+iT}= \underset{u \in [-1,1]}{\text{argmin}}\, \{\EX[V_{\epsilon}(\rho_{k+(i+1)T}) | \rho_{k+iT}, u]\} = 0 \, \forall\, i \geq K
$.
So, for $q \in \{1,\dots,T\}$, 
$
   \EX[V_{\epsilon}(\rho_{k+(i+q+1)T}) | \rho_{k+(i+q)T}, u=0] = \EX[V_{\epsilon}(\mathbb{M}_{\beta_k}^{0}\rho_{k+(i+q)T})]\leq V_{\epsilon}(\rho_{k+(i+q)T})
$
since $V_{\epsilon}$ is an open-loop supermartingale. We now know $\exists\, K \in \mathbb{N}$ such that $\EX[V_{\epsilon}(\rho_{k+T}) | \rho_k, u_k] \leq V_{\epsilon}(\rho_k) \ \forall\, k \geq K.$ Since there exists a unique mapping from elements of $D$ to elements of $\reals^{2d^2}$, we apply Theorem \ref{thm:extension} to prove $\PR[\lim_{k \to \infty} \rho_k = \ket{\bar{b}_r}\bra{\bar{b}_r}] = 1.$

Since $\ \forall a_j \in \CA, \exists$ a set $\{\ket{b_{j_1}},\dots,\ket{b_{j_n}}\} = a_j \otimes \CX$, the convergence to any $\ket{\bar{b}_r}\bra{\bar{b}_r}, \ket{\bar{b}_r} \in a_j \otimes \CX$, implies the convergence to $a_j \in \CA.$
\end{proof}

To summarize,  we showed that for the discrete time psychological state evolution of (\ref{M_update}), there exists a control policy which allows the machine to guide the human decisions such that a target decision is made asymptotically, almost surely.

\section{Finite Step Stochastic \ly Stability}
\label{sec4}
The purpose of this section is two-fold. First,
we generalize  Theorem 1 of Qin et. al. \cite{qin2019lyapunov} to the case when the finite step size can be a random variable. Our main result below is  Theorem \ref{thm:extension}. Recall that we used this result  in Section \ref{sec3} to prove stability of the human-decision system. Second, Theorem \ref{thm:rand} below is of  independent interest.

Consider the discrete time stochastic system described by 
\begin{equation}
\label{eq1}
    x_{k+1} = f(x_k, y_{k+1}), \ \ \ k= 0,1,2,\ldots
\end{equation}
Here $x_k \in \reals^n$, and $\{y_k : k =0,1,2,\ldots\}$ is a $\reals^d$ valued stochastic process on the probability space $(\Omega, \mathcal{F}, \PR)$.
Consider the filtration (increasing sequence of $\sigma$-fields)  defined by $\mathcal{F}_0 = \{\emptyset, \Omega\}, \ \mathcal{F}_k = \sigma(y_1, \dots, y_k) \textrm{ for } k \geq 1$. We choose $x_0 \in \reals^n$ as a  constant with probability one. Thus $\{x_n\}$ is an $\reals^n$-valued stochastic process adapted to $\mathcal{F}_k$. \\

\begin{theorem}
\label{thm:extension}
For the discrete-time stochastic system (\ref{eq1}), let $V: \reals^n \rightarrow \reals$ be a continuous non-negative and radially unbounded function. Suppose we have the condition:\\
(a) There exists a random variable $T: \Omega \rightarrow \mathbb{N}$ such that for any $k$, $\EX[V(x_{k+T}) | \mathcal{F}_k] - V(x_k) \leq -\varphi(x_k)$, where $\varphi: \reals^n \rightarrow \reals$ is continuous and satisfies $\varphi(x) \geq 0$ for any $x \in \reals^n$.
Then for any initial condition $x_0 \in \reals^n$, $x_k$ converges to $\mathcal{D}_1 := \{x \in \reals^n : \varphi(x) = 0\}$ with probability one. 

This Theorem follows from Theorems \ref{thm:qin} and \ref{thm:rand}; proofs are given for both of these. \\
\end{theorem}

\begin{theorem}
\label{thm:qin}
For the discrete-time stochastic system (\ref{eq1}), let $V: \reals^n \rightarrow \reals$ be a continuous non-negative and radially unbounded function. Define the set $Q_{\lambda} = \{x : V(x) < \lambda\}$ for some positive $\lambda$, and assume that:\\
\ \ \ (a) $\EX[V(x_{k+1}) | \mathcal{F}_k] - V(x_k) \leq 0$ for any $k$ such that $x_k \in Q_{\lambda}$ \\
\ \ \ (b) There exists an integer $T \geq 1$, independent of $\omega \in \Omega$, such that for any $k, \EX[V(x_{k+T}) | \mathcal{F}_k] - V(x_k) \leq -\varphi(x_k)$, where $\varphi: \reals^n \rightarrow \reals$ is continuous and satisfies $\varphi(x) \geq 0$ for any $x \in Q_{\lambda}$ \\ 
Then for any initial condition $x_0 \in Q_{\lambda}$, $x_k$ converges to $\mathcal{D}_1 := \{x \in Q_{\lambda} : \varphi(x) = 0\}$ with probability at least $1 - V(x_0)/\lambda$ \\
\end{theorem} 
\begin{proof}We have 
\begin{equation}
\label{eq2}
    \EX[V(x_{k+T})|\mathcal{F}_k] - V(x_k) \leq -\varphi(x_k) \leq 0 ,  \forall x_k \in Q_{\lambda}
\end{equation} where $\varphi(x)$ is continuous. Now, Kushner \cite[p.196]{kushner1971introduction} has shown that if we start with $x_0 \in Q_{\lambda}$ then the probability of staying in $Q_{\lambda}$ during the entire resultant trajectory is at least $1 - V(x_0)/\lambda$, i.e. 
\begin{equation}
\label{eq3}
    \PR[\sup_{k \in \mathbb{N}}V(x_k) \geq \lambda] \leq V(x_0)/\lambda
\end{equation}
Next construct  $T$ subsequences of $\{X_k\}$ as follows: $\{X_i^{(0)}\} = \{X_0,X_T,\dots\}, \{X_i^{(1)}\} = \{X_1,X_{T+1},\dots\},\dots, \{X_i^{(T-1)}\} = \{X_{T-1},X_{2T-1},\dots\}.$ \\
Suppose $\varphi(x) \geq 0 \ \forall x:$ Then for all $k \in K := \{0,\dots,T-1\}$, $\{V(X_i^{(k)})\}$ is a non-negative supermartingale process by (\ref{eq2}), and thus by Doob's convergence theorem converges to a limit with probability 1. From (\ref{eq2}) we have for all $k \in K$ and $n \in \mathbb{N}$
$\sum_{l=1}^{n}\EX(V(X_l^{(k)}))-\EX(V(X_{l-1}^{(k)}))\leq-\EX(\sum_{l=0}^{n-1}\varphi(X_l^{(k)}))
$
and $0 \leq \EX(V(X_n^{(k)})) \leq \EX(V(X_0^{(k)})) - \EX(\sum_{l=0}^{n-1}\varphi(X_l^{(k)}))$.
We use Fatou's Lemma to obtain $\EX(\sum_{l=0}^{\infty}\varphi(X_l^{(k)})) < \infty$ and by Borel-Cantelli we have $\PR[\lim_{n \to \infty}\varphi(X_n^{(k)}) = 0] = 1\, \forall\, k \in K
$.
Now suppose $\varphi(x) \geq 0$ only for $x \in Q_{\lambda}.$ Stop $\{X_n^{(k)}\}$ on first leaving $Q_{\lambda}$. Then for $x \notin Q_{\lambda}$, $\varphi(x) = 0$ for this stopped set. This stopped process is a supermartingale and the proof is the same as above.  \\
It is now apparent that $\lim_{n \to \infty}\varphi(X_n^{(k)}(\omega)) = 0 \ \forall k \in K$ and $\omega \in \bar{\Omega} = \{\omega \in \Omega : x_n(\omega) \in Q_{\lambda} \ \forall n \in \mathbb{N}\}$, so we have $\PR[\lim_{n \to \infty}\varphi(X_n(\omega)) = 0] \geq 1- V(x_0) / \lambda$ by the  analysis in  Appendix \ref{pf:subseq_conv} and (\ref{eq3}). \\
\end{proof}

\begin{theorem}
\label{thm:rand}
Theorem (\ref{thm:qin}) holds when $T$ is an integer-valued random variable $T: \Omega \rightarrow \mathbb{N}$, where $\Omega$ is the underlying sample space  \\
\end{theorem}

\begin{proof} 
This follows from the previous proof, with expectations conditioned on $T(\omega)=\tau$, the set of sequences with interval length $\tau \in \mathbb{N}$:
$
    \sum_{l=1}^{n}\EX(V(X_l^{(k)}))-V(X_{l-1}^{(k)})|T=\tau)\leq - \EX(\sum_{l=0}^{n-1}\varphi(X_l^{(k)})|T=\tau)
$. Applying Fatou's Lemma yields $\EX(\sum_{l=0}^{\infty}\varphi(X_l^{(k)})|T=\tau) < \infty$ and so by Borel-Cantelli we have $
    \PR[\lim_{n \to \infty}\varphi(X_n^{(k)}) = 0|T=\tau] = 1 \ \ \forall k \in K
$. The same arguments from the proof of Theorem (\ref{thm:qin}) yield
$
    \PR[\lim_{n \to \infty}\varphi(X_n(\omega)) = 0|T=\tau] \geq 1- V(x_0) / \lambda
$. So,
$
    \PR[\lim_{n \to \infty}\varphi(X_n(\omega)) = 0|T=\tau]\mathbb P(T=\tau) \geq (1- V(x_0) / \lambda)\mathbb P(T=\tau)
    \implies \PR[\lim_{n \to \infty}\varphi(X_n(\omega)) = 0]  = \sum_{\tau \in \mathbb{N}}\PR[\lim_{n \to \infty}\varphi(X_n(\omega)) = 0|T=\tau]\mathbb P(T=\tau) \geq \sum_{\tau \in \mathbb{N}}(1- V(x_0) / \lambda)\mathbb P(T=\tau) = 1- V(x_0) / \lambda
$
\end{proof}

To summarize, this section provided a generalization of the finite-step \ly function result of Qin. et. al. \cite{qin2019lyapunov}. We applied this to show stability of the Lindbladian dynamics to prove almost sure convergence of the density operator (psychological state), but the generalization is of independent interest. 

\section{Conclusion and Extensions}
\label{sec5}
    Our main result, Theorem (\ref{thm:convg}), showed that for a human-machine decision system modeled as a controlled quantum decision system, there exists an optimal control policy under which the human's action choice can be guided to an arbitrary target action with probability one. This can be useful in tense or stressful decision situations when the human is subject to cognitive bias and irrational preferences; The machine can act as a rational Bayesian expected utility maximizer and control the human's preference to a Bayesian optimal choice. In proving this we have utilized a novel random finite-step Lyapunov function result, which we present and prove in Sec. (\ref{sec4}), and which stands as an independent result. There are several simplifying assumptions we have made, which warrant further investigating. For one, we have assumed that the machine knows the $(\alpha,\lambda,\phi)$ parametrization of the human's Lindbladian mental operator. It would be interesting to see what the analysis yields when these are only estimates with some distribution. 
   It is worthwhile  extending our results to   more general human-machine decision systems.

\section{Appendix}

\subsection{Convergence of constructed subsequences implies convergence of sequence}
\label{pf:subseq_conv}
We have $\lim_{n \to \infty}\varphi(X_n^{(k)}(\omega)) = 0 \ \forall \omega \in  \bar{\Omega}$. Let $\omega \in \bar{\Omega}$ and $\varphi_n^{(k)}$ denote $\varphi(X_n^{(k)}(\omega))$ and $\varphi_n$ denote $\varphi(X_n(\omega))$. \\ 
We have: $\forall\, \epsilon > 0 \ \exists\, N_k$ such that $\varphi_n^{(k)} < \epsilon \ \forall n > N_k$. Take $N^{*} = \max_{k \in \{0,\dots,T-1\}}N_k$. Suppose $\lim_{n \to \infty}\varphi_n \neq 0$: Then $\exists\, \epsilon > 0$ such that $\forall \ N \in \mathbb{N} \ \exists \ n_0 > N$ with $\varphi_{n_0} > \epsilon.$ Since the subsequences are exhaustive, i.e. for any $\varphi_n \ \exists\, k, m,$ such that $\varphi_n = \varphi_m^{(k)}$, we know that for any $\epsilon > 0$, $\varphi_n < \epsilon \ \forall\, n > N^{*}$ so such a $n_0$ does not exists for $N^{*}$ and thus by contradiction we have $\lim_{n \to \infty}\varphi_n = 0.$ \\

\subsection{Lindbladian Psychological Model Construction} 
Let psychological state space $\CS:=\CH({\{\ket{\mathcal{E}_l,\CA_j}\}}_{\substack{{1\leq l\leq n}\\{1\leq j\leq m}}})$ be a Hilbert space. Here $\CH$ is defined in Notation (2). Each $\ket{\mathcal{E}_l}$ is an $n$-dimensional complex vector indexed by the state $l$, and $\ket{\mathcal{A}_j}$ is an $m$-dimensional complex vector indexed by the action $j$.
\label{constr_lind}
The evolution of the density operator is given by $
\frac{d\rho}{dt} = \mathcal{L}_{(\alpha,\lambda,\phi)} \, \rho
$ where
\begin{equation}
\label{Lindblad}
\begin{split}
     &\mathcal{L}_{(\alpha,\lambda,\phi)} \, \rho =-i(1-\alpha)[H,\rho]\\&+\alpha \sum_{m,n}\gamma_{(m,n)}\left(L_{(m,n)}\,\rho\, L_{(m,n)}^{\dagger}-\frac{1}{2}\{L_{(m,n)}^{\dagger}L_{(m,n)},\rho\}\right)    
\end{split}
\end{equation}
Here $[A,B] = AB - BA$, $\{A,B\} = AB + BA$, $A^*$ is complex conjugate of $A$,
${H}=\diag(\mathbf{1}_m,\cdots,\mathbf{1}_m)_{mn\times\,mn}$ with $\mathbf{1}_m$ an $m \times m$ matrix of ones. $L_{(m,n)} = \ket{m}\bra{n}$ is the jump operator, which represents the jump from $m^{th}$ cognitive state to $n^{th}$ cognitive state.
$\gamma_{(m,n)}:=[C(\lambda,\phi)]_{m,n}=[(1-\phi)\Pi^{T}(\lambda) + \phi K^{T}]_{m,n}$.  For utility function $u: \CA \times \CX \rightarrow \reals$,
$p(a_j|\mathcal{E}_l)=\frac{u(a_j|\mathcal{E}_l)^{\lambda}}{\sum_{j=l}^m u(a_j|\mathcal{E}_l)^{\lambda}}$
we define
$\mathbb{P}(\mathcal{E}_l):=\begin{bmatrix}
    p(a_1|\mathcal{E}_l) & p(a_2|\mathcal{E}_l)&\cdots& p(a_{m}|\mathcal{E}_l) \\
\end{bmatrix}\otimes\mathbf{1}_{n\times1}$ and 
$\Pi(\lambda) = \diag(
        P(\mathcal{E}_1),\cdots,P(\mathcal{E}_{n}))$
where $\mathbf{1}_{n\times1}$ is a vector with all 1's and, $A \otimes B$ is the kronecker product of $A$ and $B$. Define
$\eta_k(s) = p(s | u_k, y_k)$
given the noisy observation $y_k$ and input signal $u_k$, with $s \in \CX$.
We define
$K:={\begin{bmatrix}
        \eta_k(\mathcal{E}_1) & \eta_k(\mathcal{E}_2)&\cdots&\eta_k(\mathcal{E}_{n}) 
      \end{bmatrix} }\otimes\mathbf{1}_{m\times1}\otimes \mathbb{I}_{m\times m}$.
 $\alpha \in [0,1]$ represents the amount of quantum behavior in the evolution of the density operator.
    $\lambda \in [0,\infty)$ can be thought of as the agent's ability to discriminate between the profitability of different actions.
    $\phi \in (0,1)$ represents the relevance of discrimination between underlying states $\{\mathcal{E}_1, \dots, \mathcal{E}_{n}\}$.

 When we consider the closed-loop feedback control mechanism, the scalar control input $u_k$ directly parametrizes the structure of the cognitive matrix, so that we can have $C_{u_k}(\lambda,\phi)$. We leave this as a parameter and do not define the specific effect of $u_k$ on $C_{u_k}(\lambda,\phi)$, as long as $u_k$ satisfies the constraints $(i) - (iv)$ of Section~\ref{ap:contr}.
 
 \subsection{Discrete time representation of Lindblad Equation}
 \label{DT_lind}
We time  discretize the Lindbladian dynamics  by representing it using  Kraus operators to get (\ref{M_update}), see  \cite{Pearle_2012}.
We discretize the evolution in equation (\ref{Lindblad}) in multiples of $\Delta t$ in order to incorporate $T$ time steps and rewrite it as
$
    \rho_{k+T}=\sum_{\mu}M_{\mu_k,T}^{u_k}\,\rho_{k}\,M_{\mu_k,T}^{u_k \dagger}
$
where $M_{0,T}^{u_k}=\left[\mathbb{I}-T\Delta t(iH+\frac{1}{2}\sum_{m,n}\gamma_{m,n}^{u_k}L_{(m,n)}^{\dagger}L_{(m,n)})\right]$
and $M_{\mu\neq0,T}^{u_k}=\sqrt{T\Delta t\gamma_{m,n}^{u_k}}L_{(m,n)}$, $T\in\BN$. 
We then perform a measurement to get $\rho_{k+T} = \frac{P_{\nu_k}(\sum_{\mu_k}M_{\mu_k,T}^{u_k}\,\rho_{k}\,M_{\mu_k,T}^{u_k \dagger})P_{\nu_k}^{\dagger}}{\Tr(P_{\nu_k}(\sum_{\mu_k}M_{\mu_k,T}^{u_k}\,\rho_{k}\,M_{\mu_k,T}^{u_k \dagger}) P_{\nu_k}^{\dagger})}
$
and thus obtain (\ref{M_update}).
\subsection{Control Input Constraints}
\label{ap:contr}
($i$): For each $u_k, \sum_{\mu}M_{\mu}^{u_k \dagger}\, M_{\mu}^{u_k} = \id$ for general quantum measurements $\mu$.\\
($ii$): For $u_k = 0$, all $M_{\mu}^0$ are diagonal in the same orthonormal basis $\{\ket{n} |\, n \in \{1,\dots,d\}\}: M_{\mu}^0 = \sum_{n=1}^d c_{\mu,n}\ket{n}\bra{n}, c_{\mu,n} \in \complex$ \\
($iii$): For all $n_1 \neq n_2$ in $\{1,\dots,d\}$, $\,\exists\, \mu \in \{1,\dots,d\}$ such that $|c_{\mu,n_1}|^2 \neq |c_{\mu,n_2}|^2.$ \\
($iv$): All $M_{\mu}^{u_k}$ are $C^2$ functions of $u_k$

\subsection{Martingale Proof}
\label{ap:martingale}
We prove  that under the evolution of (\ref{M_update}), $\bra{b_r} \rho_k \ket{b_r}$ is a $T$-step martingale. Denote
$
    \mathbb{M}_{\mu_k}\rho_k$ $ = \mathbb{M}_{\mu_k, T=1}^0 \rho_k=\frac{M_{\mu_k, 1}^0\, \rho_k\, M_{\mu_k, 1}^{0\dagger}}{\Tr(M_{\mu_k, 1}^0\, \rho_k\, M_{\mu_k, 1}^{0\dagger})} $.
Since $\bra{b_r} \rho_k \ket{b_r} $ $= \Tr(\ket{b_r}\bra{b_r} \rho_k)$, 
$\EX[\Tr(\ket{b_r}\bra{b_r}\, \rho_{k+1}) | \rho_k,u_k]$ $=\sum_{\mu_k = 1}^m  \Tr(M_{\mu_k}\, \rho_k \, M_{\mu_k}^{\dagger}) \Tr(\ket{b_r}\bra{b_r}\, \mathbb{M}_{\mu_k}\, \rho_k)$ $=\sum_{\mu_k = 1}^m \Tr(\ket{b_r}\bra{b_r}\, M_{\mu_k}\, \rho_k\, M_{\mu_k}^{\dagger})$ $=\Tr(\sum_{\mu_k = 1}^m  M_{\mu_k} M_{\mu_k}^{\dagger} \ket{b_r}\bra{b_r} \rho_k)=\Tr(\ket{b_r}\bra{b_r} \rho_k).
$
By induction we have $\EX[\bra{b_r} \rho_{k+T} \ket{b_r}] = \bra{b_r} \rho_k \ket{b_r}$


\bibliographystyle{acm}
\bibliography{Bibliography}

\end{document}